\documentclass[letterpaper, 10 pt, conference]{ieeeconf}  

\IEEEoverridecommandlockouts                              
\usepackage{amsmath} 
\usepackage{amssymb}  
\usepackage{xcolor}
\usepackage{hyperref}
\usepackage{graphicx}
\usepackage[english]{babel}   
\usepackage{tikz}

\definecolor{myblue}{HTML}{0072BD}
\definecolor{myred}{HTML}{D95319}
\definecolor{myyellow}{HTML}{EDB120}
\definecolor{mypurple}{HTML}{7E2F8E}
\definecolor{mygreen}{HTML}{77AC30}
\usepackage{multirow}
\usepackage{aligned-overset}
\usepackage{siunitx}

\usetikzlibrary{positioning,shapes.geometric,}
\usepackage[european,americanvoltages,americancurrents]{circuitikz}

\title{\LARGE \bf
Bridging Continuous-time LQR and Reinforcement Learning \\ via Gradient Flow of the Bellman Error
}

\author{Armin Gießler, Albertus Johannes Malan and Sören Hohmann
\thanks{A. Gießler, A. J. Malan, S. Hohmann are with the Institute of Control Systems, Karlsruhe Institute of Technology (KIT), 76131, Karlsruhe, Germany. Corresponding author is Armin Gießler, 
{\tt \footnotesize armin.giessler@kit.edu}.}
}

\usepackage{fancyhdr}

\fancypagestyle{firstpage}{
\fancyhf{}

\fancyfoot[C]{\small
© 2025 IEEE. Personal use of this material is permitted. Permission from IEEE must be obtained for all other uses.\\
This is the author’s version of a work that has been published in \textit{Proceedings of the IEEE Conference on Decision and Control (CDC 2025)}. \
The final version is available at: \url{https://doi.org/10.1109/CDC57313.2025.11312822}
}
}

\newtheorem{theorem}{Theorem}
\newtheorem{proposition}{Proposition}
\newtheorem{definition}{Definition}
\newtheorem{remark}{Remark}
\newtheorem{assumption}{Assumption}
\newtheorem{lemma}{Lemma}

\newtheorem{algorithm}{Algorithm}

\usepackage[style=ieee, 
backend=bibtex,  
isbn=false,  
doi=false, 
giveninits=true,url=false,date=year, maxbibnames=99, minnames=1, maxnames=3, sorting=none,alldates=year,]{biblatex} 
\AtEveryBibitem{\clearfield{pages}}\AtEveryBibitem{\clearfield{volume}}

\begin{document}

\maketitle

\thispagestyle{firstpage}

\begin{abstract}
   In this paper, we present a novel method for computing the optimal feedback gain of the infinite-horizon Linear Quadratic Regulator (LQR) problem via an ordinary differential equation. 
   We introduce a novel continuous-time Bellman error, derived from the Hamilton-Jacobi-Bellman (HJB) equation, which quantifies the suboptimality of stabilizing policies and is parametrized in terms of the feedback gain. We analyze its properties, including its effective domain, smoothness, and coerciveness, and show the existence of a unique stationary point within the stability region. 
   Furthermore, we derive a closed-form gradient expression of the Bellman error that induces a gradient flow. This converges to the optimal feedback and generates a unique trajectory that exclusively comprises stabilizing feedback policies. 
   Additionally, this work advances interesting connections between LQR theory and Reinforcement Learning (RL) by redefining suboptimality of the Algebraic Riccati Equation (ARE) as a Bellman error, adapting a state-independent formulation, and leveraging Lyapunov equations to overcome the infinite-horizon challenge.
   We validate our method in a simulation and compare it to the state of the art.
\end{abstract}

\section{Introduction}
The Linear Quadratic Regulator (LQR) stands as a cornerstone in the realm of control theory. At its core, LQR addresses the fundamental challenge of designing a controller that minimizes a cost function while ensuring system stability. 
By applying quadratic cost functions to linear systems, LQR is known for its elegant formulation, where the optimal solution can be explicitly expressed as a state-feedback control law \cite{kalman1960}.
Over the last six decades, numerous methods have been developed to solve the LQR problem, the majority of which rely 
on obtaining the maximum solution\footnote{Maximality is considered in terms of the partial order induced by positive semidefinite matrices.} of the Algebraic Riccati Equation (ARE). 
These methods include iterative algorithms for discrete-time \cite{hewer1971iterative} and continuous-time systems \cite{kleinman1968iterative}, as well as linear-algebra-based approaches that perform eigenvalue decomposition on the Hamiltonian matrix \cite{lancaster1995algebraic}.
Additionally, the LQR problem can be transformed into a semidefinite program (SDP) \cite{balakrishnan2003semidefinite,yao1999lq}, or equivalently expressed as an $H_2$ norm minimization problem, which can be cast into a convex optimization problem \cite{chen2012optimal,feron1992numerical}.

More recently, Reinforcement Learning (RL) has offered a new perspective on learning optimal controllers for systems with nonlinear, stochastic, or uncertain dynamics \cite{sutton2018reinforcement}.
While RL has advanced rapidly with powerful algorithms to search for optimal policies\footnote{Policy is a synonym for control law.}, ensuring the stability of these learned policies remains challenging. This has led to a growing intersection between RL and control theory \cite{kiumarsi2017optimal,recht2019tour,vrabie2012optimal}. 


\subsubsection*{Literature review}
In contrast to methods relying on the ARE, an alternative approach to solving the LQR problem involves directly minimizing the infinite-horizon LQR cost that is parametrized in the feedback matrix. 
By using Lyapunov equations to handle the infinite horizon, the LQR problem can be minimized using first-order methods with a closed-form gradient.
The gradient expression was initially derived for continuous-time systems in \cite{Levine1970,GEROMEL1982545} and later adopted for discrete-time systems in \cite{maartensson2012gradient}. Rigorous analysis of this approach, including properties such as smoothness, coerciveness, and quadratic growth, is presented in \cite{fazel2018global,bu2019dlqr} for discrete-time systems and \cite{bu2020clqr,Bu2020ACC} for continuous-time systems. 

In the context of RL, policy iteration (PI) iteratively determines the optimal policy for continuous-time systems \cite{vrabie2012optimal}. This method consists of a policy evaluation step, where the continuous-time Bellman equation is used to compute the value function for a stabilizing policy, followed by a policy improvement step to update the control law. PI solves the non-linear Hamilton-Jacobi-Bellman (HJB) equation, which is a sufficient optimality condition
\footnote{For input-affine systems and stage costs, which quadratically penalize the input and contain positive definite functions in the states, the HJB equation provides both a necessary and sufficient optimality condition.}, by iterations on equations linear in the gradient of the value function. In the LQR case, the HJB equation reduces to the continuous-time ARE (CARE), and PI is exactly Kleinman's Algorithm \cite{kleinman1968iterative}. Integral RL (IRL) replaces the continuous-time Bellman equation\footnote{In contrast to the discrete-time Bellman equation, the continuous-time Bellman equation depends on the full system dynamics \cite[Sec. 2.6]{vrabie2012optimal}.} with an integral variant, enabling learning without knowledge of the drift dynamics\footnote{The drift dynamics correspond to the system matrix in the linear case.} 
or state derivatives \cite{VRABIE2009lin,VRABIE2009nonlin,vrabie2012optimal}. In IRL, the Bellman error, defined on this IRL Bellman equation and parametrized in the value matrix, is minimized during the policy evaluation step. 

All the mentioned RL methods learn the optimal value function and derive the policy from it. By directly parametrizing the learning objective in the policy, fewer parameters are generally required, potentially leading to faster convergence and reduced computational complexity in the optimization process. To the best of the authors' knowledge, no existing work parametrizes the HJB equation directly in terms of the feedback gain and formulates a continuous-time Bellman error based on this parametrization.

\subsubsection*{Contributions} 
The main contributions of the paper at hand are:
\begin{enumerate}
   \item Based on the HJB equation, we formulate a novel Bellman error for a continuous-time LQR, designed to quantify and address the suboptimality of stabilizing policies. The resulting matrix function parametrizes the error directly in terms of the feedback gain matrix.
   \item We rigorously analyze the  properties of the Bellman error by addressing its effective domain and proving smoothness, coerciveness, and the existence of a unique stationary point.
   \item We derive a closed-form expression for the gradient of the Bellman error and establish its well-definedness within the region of stabilizing policies. 
    Furthermore, we show that the corresponding gradient flow generates a unique trajectory of stabilizing feedbacks and converges to the solution of the LQR problem.
\end{enumerate}

\subsubsection*{Paper Organization}
This section ends with some notation and preliminaries.
Section \ref{sec:2} introduces the novel Bellman error formulation and examines its  properties. In Section \ref{sec:grad_flow}, the gradient expression and the corresponding gradient flow are presented.
Section \ref{sec:3} provides simulation results and compares our method to a state-of-the-art approach.
Finally, the paper ends with a conclusion in Section \ref{sec:conc}.

\subsubsection*{Notation}
\label{subsubsec:notation}
We denote the set of (nonnegative) real numbers as $\mathbb{R}~ (\mathbb{R}_{\geq 0})$. The transpose of a matrix $A$ (a vector $B$) is denoted by $A^\top$ $(B^\top)$. The identity matrix of dimensions $n\times n$ is given by $I_n$. A symmetric positive definite (semidefinite) matrix $A$ is denoted by $A\succ 0 $ $(A \succeq 0)$. The trace of a square matrix $A$ is denoted by $\operatorname{tr }(A )$.
The spectral (Frobenius) norm of a matrix $A$ is denoted by $\Vert A \Vert_2$ $(\Vert A \Vert_F)$.
The spectrum of a square matrix $A$, containing the eigenvalues $\lambda_i(A)$, is denoted by $\sigma (A)$. The minimum eigenvalue of $A$ is represented by $\lambda_{\min}(A)$. The vectorization and rank of a matrix $A$ are denoted by $\operatorname{vec}(A)$ and $\operatorname{rank}(A)$, respectively. The gradient of a scalar function $f$ with respect to a vector $x$ (matrix $A$), denoted by $\nabla_x f$ $(\nabla_A f)$, is a column vector (matrix) of partial derivatives. 
A square matrix $A$ is Hurwitz stable if its eigenvalues have strictly negative real parts. 

\subsection{Preliminaries}
\subsubsection{Mathematical Preliminaries}
The continuous-time Lyapunov equation is given by 
\begin{align}
   \label{math:Lyap0}
   A P + P A^\top + L = 0, 
\end{align}
where $A,P,L \in\mathbb{R}^{n \times n}$ and $P,L$ are symmetric. 
If $A$ is Hurwitz stable, the solution of \eqref{math:Lyap0} can be written as $   P = \int_0^\infty e^{A t }L e^{A^\top t }\mathrm{d}t < \infty$.

The vectorization of the matrix product $AB$, where $A\in\mathbb{R}^{k\times l},B\in\mathbb{R}^{l\times m}$, is given by
\begin{align*}
   \operatorname{vec}(AB)= \left(I_m \otimes A  \right)\operatorname{vec}(B  ) = \left(B^\top \otimes I_k  \right)\operatorname{vec}(A).
\end{align*}
The inverse of the vectorization operator $\operatorname{vec}^{-1}_{k\times l}: \mathbb{R}^{k l}\to \mathbb{R}^{k\times l}$ can be defined as
\begin{align}
   \label{math:inv}
   A\mapsto  \operatorname{vec}^{-1}_{k\times l}(A) = \left( \operatorname{vec}^\top(I_l ) \otimes I_k\right)(I_l \otimes A).
\end{align}
The Lyapunov equation \eqref{math:Lyap0} can be solved using vectorization\footnote{Rewriting the equation in this form can be computationally expensive and ill-conditioned. A standard approach to solve the Lyapunov equation is the Bartels-Stewart algorithm \cite{bartels1972solution}.}
\begin{align}
   \label{math:vec0}
   \operatorname{vec}(P) & = - \left( I_n \otimes A + A^\top \otimes I_n  \right)^{-1}\operatorname{vec}(L).
\end{align}

Since the trace operator is a linear operator, $ \operatorname{tr }(A+B) = \operatorname{tr }(A ) + \operatorname{tr }(B)$ and $\operatorname{tr }(c A )= c \operatorname{tr }( A )$ hold. Additionally, the trace is invariant under circular shifts and transposition, i.e., $\operatorname{tr }(ABC ) = \operatorname{tr }(BCA ) = \operatorname{tr }(CAB )$ and $ \operatorname{tr }(A )  = \operatorname{ tr }(A^\top)$. The trace of a square matrix $A$ can be calculated as $\operatorname{tr }(A )= \sum_{i=1 }^{n }\lambda_i(A)$.
\subsubsection{Linear Time-Invariant Systems}
We consider a continuous-time linear time-invariant (LTI) system
\begin{align}
   \dot x(t) &= A x(t) + B u(t), \\
   y(t) &= C x(t),
\end{align}
where $A\in\mathbb{R}^{n\times n},B\in\mathbb{R}^{n\times m}$, and $C\in\mathbb{R}^{q\times n}$. 
The pair $(A,B)$ is stabilizable if and only if, for every $\lambda \in \sigma(A)$ with $\operatorname{Re}(\lambda) \geq 0$, the matrix $\begin{bmatrix} \lambda I_n - A & B \end{bmatrix}$ has full rank. Similarly, the pair $(A,C)$ is detectable if and only if, for every $\lambda \in \sigma(A)$ with $\operatorname{Re}(\lambda) \geq 0$, the matrix $\begin{bmatrix} \lambda I_n - A & C \end{bmatrix}$ has full rank.
There exists a time-invariant state feedback controller such that $A-BK$ is Hurwitz stable if and only if the system $(A,B)$ is stabilizable \cite[Theorem 14.5]{hespanha2018}. The set of Hurwitz stabilizing feedback gains is denoted by 
\begin{align}
   \mathcal{K} = \{  K \in\mathbb{R}^{m\times n }: A-BK \text{ is Hurwitz stable}\}.
\end{align}
The set $\mathcal{K}$ is in general open, unbounded, and path-connected \cite[Section 3]{bu2019}.
\subsubsection{Linear Quadratic Regulator}
We briefly recall the standard continuous-time LQR problem which is solved in the remainder of this work \cite{anderson2007}.  
\begin{definition}[LQR Problem]
The LQR problem is defined as
   \label{def:LQR}
      \begin{align}
         \label{math:LQR}
         \begin{split}
         \min_{x,u}~ & \int_{0}^{\infty} u^\top(t) R u(t) + x^\top(t) Q x(t) \mathrm{d}t, \\
         \text{s.t.}~ & \dot x(t) = A x(t)+B u(t), \\
         & x(0) = x_0,
         \end{split}
      \end{align}
   where $Q\succeq 0, R\succ 0$ are weighting matrices of appropriate dimensions.
\end{definition}

If the system $(A,B)$ is stabilizable and the pair $(A,\sqrt{Q})$
is detectable, the optimal value $V^*(x_0)$ of the optimization problem is finite and the optimal solution\footnote{In the following, we omit the explicit time dependency $t$ of variables for simplicity, except where the dependency on $t$ is relevant.} $u^*=-K^* x$ stabilizes the system $(A,B)$. 
The time-invariant optimal feedback, often referred to as the Kalman gain, is given by $K^*=R^{-1}B^\top P^*$, where the Riccati matrix $P^*$ is the positive definite solution of the CARE
\begin{align}
   \label{math:CARE}
   A^\top P + PA - PB R^{-1} B^\top P +Q = 0. 
\end{align}
The optimal value function of the LQR Problem is given by $V^*(x_0)=x_0^\top P^* x_0$. We note that the CARE generally has multiple solutions. Under the stabilizability and detectability conditions, there is 
only one unique positive definite stabilizing solution such that the matrix $A-BR^{-1}B^\top P $ is Hurwitz stable. 
By substituting $K=R^{-1}B^\top P$, \eqref{math:CARE} can be rewritten as
\begin{align}
   A_K^\top P + P A_K + Q + K^\top R K  &=0,  \label{math:Ly}
\end{align}
where $A_K := A-BK$. Note that \eqref{math:Ly} has the form of a Lyapunov equation if $K$ is considered independent of $P$. 
We highlight that the optimal feedback cannot be obtained from \eqref{math:Ly} since there are two unknown matrices $P$ and $K$. For a stabilizing feedback $K$, the corresponding value matrix $P$ can be calculated by solving the Lyapunov equation \eqref{math:Ly}.
\subsubsection{Policy Gradient Flow of the LQR}
The cost of the LQR problem \eqref{math:LQR} can be parametrized with respect to $K$ as
\begin{align}
   \label{math:LQR_cost} f_K:\mathbb{R}^{m\times n}\to\mathbb{R}, \quad
   f_K = x_0^\top P_K x_0 = \operatorname{tr}(P_K x_0 x_0^\top),
\end{align}
where $P_K$\footnote{The subscript $K$ of matrices emphasize that the matrix depends on feedback gain $K$. An exception is the definition $A_K = A-BK$.} is the solution of \eqref{math:Ly} for a given (non-optimal) feedback $K$ \cite[Section 3.2]{bu2020clqr}. The cost $f_K$ is minimized by the gradient flow \cite[Lemma 4.1]{bu2020clqr}
\begin{align}
   \dot K(t) &= - \nabla_K f_K(t), \quad K(0) = K_0 \in\mathcal{K}   \label{math:grad}   
\end{align}
where 
\begin{align}
   \nabla_K f_K & = 2 \left( R K - B^\top P_K \right) Y_K, \\
   Y_K & = \int_0^\infty e^{A_K t} x_0 x_0^\top e^{A_K^\top t}\mathrm{d}t.
\end{align}
The matrix $Y_K$ can be obtained by solving the matrix Lyapunov equation $A_K Y_K + Y_K A_K^\top + x_0 x_0^\top = 0$. The gradient flow \eqref{math:grad} converges to $K^*$ if and only if the initial feedback $K_0$ stabilizes the system $(A,B)$, i.e., $K_0\in\mathcal{K}$. 
The trajectory $K(t)$ generated by \eqref{math:grad} always remains within the set of stabilizing feedbacks $\mathcal{K}$. 

The set $\mathcal{K}$ forms a Riemmanian manifold endowed with the inner product $\langle \cdot , \cdot \rangle_{Y_K} \!= \!\Vert  \cdot  \Vert_{Y_K}$.
The natural gradient flow
\begin{align}
 \dot{K}(t) = - \bigl( \nabla_K f_K(t) \bigr) Y_K^{-\gamma}, \label{math:natural} 
\end{align}  
where $\gamma>0$, follows the direction of the steepest descent with respect to the Riemmanian geomerty \cite[Section 5]{Bu2020ACC}.

\subsubsection{Reinforcement Learning}
The LQR problem in Definition \ref{def:LQR} is a special case of continuous-time Reinforcement Learning for input-affine systems \cite[Chapter 3]{vrabie2012optimal}. For this work, we focus exclusively on the LQR case and present the relevant formulas.
The value function of a stabilizing policy $u=\mu(x)$ is defined as
\begin{align}
   V^\mu (x_0)=\int_0^\infty x^\top  Q x + \mu^\top(x) R \mu(x)   \mathrm{d}t  <\infty, \label{math:value}
\end{align}
and can be seen as a cost-to-go from state $x_0$ at time instant $t$ under the policy $\mu(x)$. By using Leibniz' formula, the infinitesimal version of \eqref{math:value} is 
\begin{align}
   0 = \underbrace{x^\top  Q x + \mu^\top(x) R \mu(x) + \left(\nabla_x V^\mu(x)\right)^\top\dot x}_{=:H(x,\mu(x),\nabla_x V^\mu)}, \label{math:BE}
\end{align}
where $V^\mu(0)=0$. Equation \eqref{math:BE} can be considered as a Bellman equation for continuous-time systems\footnote{
In contrast to the discrete-time Bellman equation, $V(x_k)=x_k^\top Q x_k + u_k^\top R u_k + V(x_{k+1})$, \eqref{math:BE} depends on the system dynamics $\dot x$.}. The optimal value function $V^*(x)$ satisfies the Hamilton-Jacobi-Bellman (HJB) equation 
\begin{align}
   0 = \inf_{\mu} H(x,\mu(x),\nabla_x V^*), \label{math:HJB}
\end{align}
which reduces to the CARE \eqref{math:CARE} in the LQR case.
The following policy iteration algorithm solves the nonlinear HJB equation \eqref{math:HJB} by iterating over equations which are linear in the gradient of the value function.
\begin{algorithm}[Policy iteration \texorpdfstring{\cite[Algorithm 4.1]{vrabie2012optimal}}{}]
\label{alg:PI}
Select a stabilizing policy $\mu_0$. 
\begin{enumerate}
   \item Policy evaluation step: Solve Bellman equation for $V^{\mu_i}$
   \begin{align}
      0 &= x^\top Q x + \mu_i^\top R \mu_i + \left(\nabla_x V^{\mu_i}(x)\right)^\top (Ax + B \mu_i).
   \end{align}
   \item Policy improvement step: Update the policy using 
   \begin{align}
      \mu_{i+1} &= \arg \min_{\mu} H(x,\mu(x),\nabla_x V^{\mu_i}), \\
      & = - \tfrac{1}{2} R^{-1} B^\top \nabla_x V^{\mu_i}. \label{math:control}
   \end{align}
\end{enumerate}
\end{algorithm}

Algorithm \ref{alg:PI} is guaranteed to converge to the optimal policy, but has the drawback that full knowledge of the system dynamics, i.e. $(A,B)$, is needed \cite{leake1967}. For the LQR, Algorithm \ref{alg:PI} is equivalent to Kleinman's Algorithm \cite{kleinman1968iterative}, where the two steps are
\begin{subequations} 
   \label{math:Kleinman}
\begin{align}
   A_{K_i}^\top P_i + P_i  A_{K_i} &= - (K_i^\top R K_i + Q ), \\
   K_{i+1} &= R^{-1} B^\top P_i.  
\end{align}
\end{subequations}

\section{Problem Formulation and analytic properties}
\label{sec:2}
In this section, we define our Bellman error and study its analytic properties. To ensure that the LQR optimization problem is well-posed and that the optimal solution yields a stabilizing controller, we impose the following assumption.
\begin{assumption}
   \label{ass:1}
The system $(A,B)$ is stabilizable
 and the pair $(A,\sqrt{Q})$ is detectable. Additionally, $B \neq 0$ and $Q\neq 0$.
\end{assumption}

Throughout the remainder of this paper, we assume that Assumption \ref{ass:1} holds.
An improved policy for a given value function $V^\mu(x) = x^\top P x, P\succ 0 $ can be calculated by \eqref{math:control}\footnote{In general, there exist infinitely many policies for a value matrix $P$ if $P\neq P^*$. }
\begin{align}
   \mu &= - R^{-1}B^\top P x. \label{math:control1}
\end{align}
By inserting \eqref{math:control1} into the HJB equation \eqref{math:HJB}, we obtain
\begin{align}
   0 & = x^\top (\underbrace{ 2 P A - P B R^{-1} B^\top P + Q}_{=: \tilde M}   )x = \operatorname{tr }(\tilde M x x^\top ) \label{math:CARE2}, \\
   & = x^\top \! (\underbrace{A^\top P + P A - \! P B R^{-1} B^\top P + Q}_{=:  M}) x = \operatorname{tr }(M x x^\top \!), \label{math:CARE3} 
\end{align}
where $M=0$ corresponds to the CARE.

\begin{remark}
Only the symmetric part $\frac{1}{2}( \tilde M^\top + \tilde M )=M$ contributes to the quadratic form and the trace in \eqref{math:CARE2}. Therefore, \eqref{math:CARE2} and \eqref{math:CARE3} are equivalent.
\end{remark}

It is well known that the optimal feedback gain $K$ of the LQR problem \eqref{math:LQR} is state-independent. Inspired by the initial state-independent formulation of \cite{bu2019dlqr,bu2020clqr}, we reformulate \eqref{math:CARE3} as
\vspace{-0.3cm}
\begin{align}
   0 = \sum_{j=1}^n \operatorname{tr }(M x_j x_j^\top), \label{math:ind}
\end{align}
where the states $x_1,\dots,x_n$ are linearly independent.
Without loss of generality, we may choose $x_1,\dots,x_n$ as the standard basis vectors, in which case $\sum_{j=1}^n x_j x_j^\top = I_n$, and \eqref{math:ind} simplifies to $0=\operatorname{tr }(M)$ by linearity of the trace operator.
The equation \eqref{math:ind} holds if $P=P^*$, since $M=0$.\footnote{Note that, more generally, any matrix $M$ for which the sum of its eigenvalues is zero satisfies \eqref{math:ind}, not only $M=0$.} In order to quantify the suboptimality of a given stabilizing feedback, we make the following definition.

\begin{definition}[Bellman error]
   \label{def:sbe}
   The Bellman error param-etrized in the feedback $K$ is defined as the matrix function $e_K:\mathbb{R}^{m\times n}\to \mathbb{R}$
   \begin{align}
      \label{math:sbe}
       K \mapsto e_K = -\operatorname{tr }(\underbrace{A^\top P_K + P_K A - P_K B R^{-1} B^\top P_K + Q}_{=:M_K}),
   \end{align}
   where the matrix $P_K$ satisfies the Lyapunov equation
   \begin{align}
      0 & = A_K^\top P_K + P_K A_K + Q + K^\top  R K. \label{math:Ly2}
   \end{align}
\end{definition}

\subsection{Analytic Properties of the Bellman Error}
We now investigate for which $K\in\mathbb{R}^{m\times n}$ the Lyapunov equation \eqref{math:Ly2} has a unique solution.
\begin{lemma}
   \label{lemma:1}
   The Lyapunov equation \eqref{math:Ly2} has a unique solution $P_K$ if and only if $A_K$ and $-A_K$ have no eigenvalue in common, i.e., $\sigma(A_K)\cap  \sigma(-A_K)=\emptyset$.
\end{lemma}
\begin{proof}
   The Lyapunov equation \eqref{math:Ly2} is a special case of the Sylvester equation $A X -XB = C$. By Sylvester equation Theorem, \eqref{math:Ly2} has a unique solution if and only if  $\sigma(A_K)\cap  \sigma(-A_K)=\emptyset$ \cite[Theorem 2.4.4.1]{horn2012matrix}.
\end{proof}

The set of feedback matrices that satisfy the condition of Lemma \ref{lemma:1} is denoted by 
\begin{align}
   \mathcal{K}_\sigma := \{ K\in\mathbb{R}^{n\times m}: \sigma(A_K)\cap  \sigma(-A_K)=\emptyset\}. \label{math:K_sigma}
\end{align}

\begin{remark}
   \label{remark:inv_vec}
   The Lyapunov equation \eqref{math:Ly2} can be solved using vectorization for $K\in\mathcal{K}_\sigma$
   \begin{align}
      \label{math:vec}
      \operatorname{vec}(P_K  )&= - \left(I_n \otimes A_K^\top + A_K^\top \otimes I_n \right)^{-1} \operatorname{vec}(Q + K^\top R K   ). 
   \end{align}
   By applying the inverse vectorization operator \eqref{math:inv}, the Bellman error \eqref{math:sbe} can be expressed in closed form for $K\in\mathcal{K}_\sigma$. 
\end{remark}

\begin{lemma}
 If $A_K$ is Hurwitz stable, then the solution $P_K$ of the Lyapunov equation \eqref{math:Ly2} is unique and positive semidefinite.
\end{lemma}
\begin{proof}
   Since $A_K$ is Hurwitz stable, its eigenvalues have strictly negative real parts, ensuring $\sigma(A_K)\cap  \sigma(-A_K)=\emptyset$. By Lemma \ref{lemma:1}, \eqref{math:Ly2} has a unique solution $P_K$. Moreover, the Hurwitz stability of $A_K$ implies that $P_K$ can equivalently be represented as
   \begin{align}
      P_K = \int_{0}^{\infty} e^{A_K^\top t } \left(Q + K^\top R K \right) e^{A_K t }\mathrm{d}t. 
   \end{align}
   Since $Q\succeq 0$ and $R\succ 0$, $Q+K^\top R K \succeq 0$. Therefore, the integrand is positive semidefinite for all $t\geq 0$, leading to $P_K\succeq 0$.
\end{proof}

\begin{lemma}
   \label{lemma:M_K_def}
   The matrix $M_K$ of Definition \ref{def:sbe} is negative semidefinite for $K\in\mathcal{K}_\sigma$. 
   Furthermore, the image of the matrix function $e_K$ over $\mathcal{K}_\sigma$ satisfies $e_K(\mathcal{K}_\sigma)\subseteq \mathbb{R}_{\geq 0}$.
\end{lemma}
\begin{proof}
   In the following, we find an equivalent expression for $M_K$ by using the Lyapunov equation \eqref{math:Ly2}.
   By substituting $A_K=A-BK$ into \eqref{math:Ly2} and solve for $A^\top P_K + P_K A$, we obtain
   {\small
   \begin{align}
      A^\top P_K + P_K A  = - Q - K^\top R K + K^\top B^\top P_K + P_K B K \label{math:MK_temp} 
   \end{align}}Inserting \eqref{math:MK_temp} into $M_K$ yields   {\small
   \begin{align}
      M_K &= -K^\top R K + K^\top B^\top P_K + P_KBK - P_K B R^{-1} B^\top P_K, \notag \\
      & = -\left( K - R^{-1}B^\top P_K \right)^\top R \left(  K - R^{-1}B^\top P_K\right). \label{math:MK_new}
   \end{align}}Since $R\succ 0$ and $P_K$ is well-defined and finite for $K\in\mathcal{K}_\sigma$ (see Lemma \ref{lemma:1}), the matrix $M_K$ is negative semidefinite \cite[Observation 7.1.8]{horn2012matrix}. Moreover, since $M_K$ has only nonpositive eigenvalues, $-\operatorname{tr }(M_K) = -\sum_i \lambda_i(M_K) \geq 0$ and thus $e_K(\mathcal{K}_\sigma)\subseteq \mathbb{R}_{\geq 0}$ holds.
\end{proof}

The following four lemmas are closely related to the results presented in \cite{bu2020clqr,bu2019dlqr}, where similar conclusions were drawn for the continuous-time and discrete-time LQR cost. We adapt these results to the Bellman error \eqref{math:sbe}.
In the next lemma, we study for which feedback $K$ the Bellman error is well-defined and address the effective domain $\operatorname{dom}(e_K)=\{ K\in\mathbb{R}^{m\times n}: e_K < \infty \}$.

\begin{lemma}
   \label{lemma:2}
   The effective domain of the matrix function $e_K$ \eqref{math:sbe} is 
   given by $\operatorname{dom}(e_K) = \mathcal{K}_\sigma$.
\end{lemma}
\begin{proof}
   Through Lemma \ref{lemma:1}, the Lyapunov equation \eqref{math:Ly2} has a unique solution $P_K$ if and only if $K\in\mathcal{K}_\sigma$. When $P_K$ exists, $M_K$ is well-defined as $M_K = A^\top P_K + P_K A - P_K B R^{-1} B^\top P_K + Q$, and consequently $e_K = -\operatorname{tr}(M_K)$ is finite. If $K\notin\mathcal{K}_\sigma$, the Lyapunov equation does not admit a solution. Specifically, $P_K$ grows unbounded, leading to $e_K \to \infty$. Therefore, $e_K$ is well-defined and finite if and only if $K\in\mathcal{K}_\sigma$.
\end{proof}

\begin{remark}
   In contrast to \cite{bu2020clqr}, where the set of stabilizing feedback gains $\mathcal{K}$ is the interior of $\operatorname{dom}(f_K)$, the effective domain of $e_K$ contains unstable feedbacks, i.e., $K\in\mathcal{K}_{\sigma}\setminus \mathcal{K}$.
\end{remark}

\begin{lemma}
   \label{lemma:3}
   The Bellman error \eqref{math:sbe} is real analytic over $\mathcal{K}_{\sigma}$, i.e., $e_K\in C^\omega(\mathcal{K}_{\sigma})$.\footnote{The class $C^\omega$ denotes smooth functions that can be locally represented by convergent power series.}
\end{lemma}
\begin{proof}
   For every $K\in\mathcal{K}_{\sigma}$, the matrix $I_n \otimes A_K^\top + A_K^\top \otimes I_n $ (see \eqref{math:vec}) is invertible, since the eigenvalues of it are $\{\lambda_i(A_K)+\lambda_j(A_K) \mid 1\leq i,j \leq n \}$ and hence nonzero. By Cramer's rule, the map $K \mapsto P_K$ \eqref{math:vec} is a rational function of polynomials in the entries of $K$ and thus $C^\omega$. The matrix function $e_K$ is a composition of $C^\omega$ maps 
   \begin{align}
      \label{math:comp}
      K \mapsto P_K \mapsto M_K \mapsto -\operatorname{tr}(M_K)
   \end{align}
   and thus real analytic over $\mathcal{K}_{\sigma}$.
\end{proof}

Since the Bellman error $e_K$ is real analytic over $\mathcal{K}_{\sigma}$, it is infinitely differentiable and all derivatives are continuous. 
For the remainder of this section, we analyze $e_K$ over the set of stabilizing feedbacks $\mathcal{K}\subset \mathcal{K}_{\sigma}$.
We observe that $e_K$ diverges to infinity either as $K$ approaches the boundary of $\mathcal{K}$\footnote{For LTI systems with $m=1$, the equality in \eqref{math:boundary} holds. However, for $m>1$, establishing whether equality holds or only a subset relation applies becomes significantly more intricate (see \cite[Prop. 3.5]{bu2019}).}
\begin{align}
   \partial \mathcal{K} \subseteq \{K\in\mathbb{R}^{m\times n}: \max_{\lambda \in \sigma(A_K)} \operatorname{Re}(\lambda) =0 \}, \label{math:boundary}
\end{align}
or when $K\in\mathcal{K}$ becomes unbounded.
\begin{lemma}
   \label{lemma:4}
   The Bellman error $e_K$ is coercive over $\mathcal{K}$ in the sense that 
   \begin{align}
      \lim_{K_i \to K\in\partial \mathcal{K}} e_{K_i} = \infty   ~ \text{ and} ~ \lim_{\Vert K\Vert_2 \to \infty,K\in\mathcal{K}}e_K=\infty.
   \end{align}
\end{lemma}
\vspace{0.15cm}
\begin{proof}
   Let $\{K_i\}$ be a sequence in $\mathcal{K}$ converging to a fixed $K\in\partial \mathcal{K}$. The corresponding sequence of matrices $\{P_{K_i}\}$ satisfies \eqref{math:vec}. As $K_i\to K\in\partial\mathcal{K}$, the matrix $I_n\otimes A_{K_i}^\top + A_{K_i}^\top \otimes I_n $ has at least one eigenvalue approaching zero and thus $\Vert \left( I_n\otimes A_{K_i}^\top + A_{K_i}^\top \otimes I_n  \right)^{-1}\Vert_2 \to \infty$. Since $\operatorname{vec}(Q+K_i^\top R K_i )$ remains bounded for all $K_i$, it follows that $\Vert P_{K_i}\Vert_2\to\infty$ as $K_i\to K \in\partial \mathcal{K}$. 
   On the other hand, consider $\Vert K \Vert_2 \to \infty$ with $K\in\mathcal{K}$. 
   Taking the spectral norm on the Lyapunov equation \eqref{math:Ly2} yields
\begin{align*}
   \Vert (A-BK)^\top P_K + P_K ( A-BK) \Vert_2 = \Vert -  Q - K^\top R K \Vert_2.
\end{align*}
The left-hand side can be upper-bounded with the triangle inequality and submultiplicativity as 
\begin{align}
   \Vert (A-BK)^\top P_K + P_K ( A-BK) \Vert_2 \\ \leq 2 \Vert A \Vert_2 \Vert P_K \Vert_2 + 2 \Vert B \Vert_2 \Vert K \Vert_2 \Vert P_K \Vert_2
\end{align}
and a lower bound for the right-hand side can be obtained by using the reverse triangle inequality\footnote{The reverse triangle inequality state that $\Vert A-B \Vert \geq \vert \Vert A \Vert - \Vert B \Vert \vert$.} and the inequality $ \vert \Vert A \Vert - \Vert B \Vert \vert \geq \Vert B \Vert - \Vert A \Vert $\footnote{The inequality holds universally due to the property of absolute values $\vert x \vert \geq \pm x$ for all $x\in\mathbb{R}$. Specifically, letting $x=\Vert A \Vert - \Vert B \Vert$, we have $\vert x \vert \geq - x$, which simplifies to $\vert \Vert A \Vert - \Vert B \Vert \vert \geq \Vert B \Vert - \Vert A \Vert$. Equality occurs when $\Vert B \Vert \geq \Vert A \Vert$.} 
\begin{align}
   \Vert - Q - K^\top R K \Vert_2 & \geq \vert \Vert -Q \Vert_2 - \Vert K^\top R K \Vert_2 \vert \\
   & \geq \Vert K^\top R K \Vert_2  -\Vert -Q \Vert_2 \\
   & \geq \lambda_{\min}(R) \Vert K \Vert_2^2 - \Vert  Q \Vert_2.
 \end{align}
Combining these bounds results in 
\begin{align*}
   2\left(  \Vert A \Vert_2 +  \Vert B \Vert_2 \Vert K \Vert_2  \right) \Vert P_K \Vert_2 &\geq \lambda_{\min}(R) \Vert K \Vert_2^2 - \Vert  Q \Vert_2, \\
   \Leftrightarrow  \Vert P_K \Vert_2 &\geq \frac{\lambda_{\min}(R) \Vert K \Vert_2^2 - \Vert  Q \Vert_2}{ 2 \Vert A \Vert_2 + 2 \Vert B \Vert_2 \Vert K \Vert_2}.
\end{align*}
Therefore $\Vert P_K \Vert_2 \to \infty$ as $\Vert K \Vert_2 \to \infty$ with $K\in\mathcal{K}$.
Next, we observe that 
\begin{align*}
   \Vert M_{K} \Vert_2 &\geq \vert \Vert A^\top P_K + P_{K}A + Q \Vert_2 - \Vert P B R^{-1} B^\top P \Vert_2 \vert\\
   & \geq \Vert P_K B R^{-1} B^\top P_K \Vert_2 - \Vert A^\top P_K + P_{K}A + Q \Vert_2  \\
   & \geq \Vert P_K B R^{-1} B^\top P_K \Vert_2 - \Vert 2 P_K A \Vert_2 - \Vert Q \Vert_2 \\
   & \geq \lambda_{\min}^+(BR^{-1} B^\top )  \Vert P_{K} \Vert_2^2 - 2\Vert A \Vert_2 \Vert P_{K} \Vert_2 - \Vert Q \Vert_2, 
\end{align*}
where $\lambda_{\min}^+(\cdot)$ denotes the smallest non-zero eigenvalue of a matrix.
For sufficiently large $\Vert P_K \Vert_2$, the quadratic term dominates the other two terms, leading to $\Vert M_K \Vert_2 \to \infty$ as $\Vert P_K \Vert_2\to\infty$. Since $M_K \preceq 0$ (see Lemma \ref{lemma:M_K_def}), its eigenvalues are non-positive, ensuring $e_K = - \operatorname{tr}(M_K)  \geq 0$. Given that $B\neq 0$ (see Assumption \ref{ass:1}), it follows that $e_K\to \infty$ as $\Vert P_K \Vert_2 \to \infty$. Combining these results, we conclude that $e_K\to\infty$ as $K_i \to K\in\partial \mathcal{K}$ or $\Vert K \Vert_2 \to \infty$.
\end{proof}

With coercivity established, we now analyze the compactness of the sublevel sets of $e_K$ over $\mathcal{K}$.

\begin{lemma}
\label{lemma:5} 
For every $\alpha>0$, the sublevel set of the Bellman error $e_K$ over $\mathcal{K}$,
\begin{align}
   L_\alpha = \{K\in \mathcal{K}: e_K \leq \alpha \},
\end{align}
is compact.
\end{lemma}

\begin{proof}
   To prove that the sublevel set $L_\alpha$ is compact, we establish that $L_\alpha$ is both closed and bounded. By Lemma~\ref{lemma:3}, $e_K$ is real analytic over $\mathcal{K}_{\sigma}$ and hence continuous over $\mathcal{K}\subset \mathcal{K}_\sigma$. The sublevel set $L_\alpha$ is the preimage of the closed interval $(-\infty,\alpha]$ under $e_K$. Continuity of $e_K$ ensures that $L_\alpha$ is closed in $\mathcal{K}$ \cite[Theorem 18.1]{munkres2014topology}. Lemma \ref{lemma:4} establishes that $e_K\to\infty$ as $\Vert K \Vert_2 \to\infty $ or $K \to \partial K\in \mathcal{K}$. Suppose now that $L_\alpha$ is unbounded. Then there exists a sequence $\{K_n\}\subset L_\alpha$ with $\Vert K_n \Vert_2 \to \infty$ or $K_n\to K\in\partial\mathcal{K}$. By coercivity, $e_{K_n}\to\infty$, which contradicts $e_{K_n}\leq \alpha$. Thus, $L_\alpha$ is bounded. Since $L_{\alpha}$ is both closed and bounded, it is compact.
\end{proof}

\section{Gradient Flow of the Bellman Error}
\label{sec:grad_flow}
In this section, we study the behavior of the gradient flow of the Bellman error.
Since the Bellman error $e_K$ is real analytic over $\mathcal{K}$ (see Lemma \ref{lemma:3}), the gradient is well-defined at all points within this domain. The closed-from expression of the gradient is provided in the following proposition. 
\begin{proposition}
   For $K\in\mathcal{K}$, the gradient of the Bellman error $e_K$ is given by\footnote{We note that the gradient is given in denominator layout notation such that $\nabla_K e_K \in \mathbb{R}^{m\times n}$ has the same dimensions as $K$.}
   \begin{align}
      \nabla_K e_K & = -4  \left(R K - B^\top P_K  \right) X_K, \label{math:grad2}
   \end{align}
   where 
   \begin{align}
      0& =A_K X_K + X_K A_K^\top + \frac{1}{2}\left(\tilde{A}_K + \tilde{A}_K^\top \right), \label{math:Lyap_A}  \\
      \tilde{A}_K & := A - B R^{-1} B^\top P_K. 
   \end{align}
\end{proposition}
\vspace*{0.2cm}
\begin{proof}
   The Bellman error is a composition of $C^\omega$ maps, $K\mapsto P_K \mapsto M_K \mapsto -\operatorname{tr}(M_K)$. 
   The differential of \eqref{math:Ly2} with respect to $K$ is given by 
   \begin{align}
      0 &= (-B dK)^\top P_K + A_K^\top dP_K \notag\\
      & ~~ + ~ dP_K A_K + P_K (-B dK) + dK^\top RK + K^\top R dK \notag \\
 &=A_K^\top dP_K + dP_K A_K^\top +  \underbrace{dK^\top \left(RK - B^\top P_K \right)}_{=:U} \notag\\ 
   &~~ + ~ \left(K^\top R - P_KB \right)dK  \\
    & = A_K^\top dP_K + dP_K A_K^\top + U + U^\top. \label{math:lyap_temp}
  \end{align} 
Since $A_K$ is Hurwitz stable for every $K\in\mathcal{K}$, the solution of the Lyapunov equation \eqref{math:lyap_temp} is unique and can be expressed as  
   \begin{align}
    dP_K = \int_0^\infty e^{A_K^\top t}(U+U^\top) e^{A_K t}\mathrm{d }t. \label{math:lyap3}
   \end{align}
   The differential of $M_K$ \ref{math:sbe} is given by
  \begin{align}
   \label{math:dMK}
   \begin{split}
   dM_K &= A^\top dP_K + dP_K A \\
   & ~~ - dP_K B R^{-1} B^\top P_K - P_K B R^{-1} B^\top dP_K.
   \end{split}
  \end{align}
  Putting \eqref{math:dMK} into the differential of Bellman error yields
  { 
  \begin{align}
   d e_K & = - \operatorname{tr}(dM_K)\\
   & =  - 2\operatorname{tr}(dP_K A )+2 \operatorname{tr}(d P_K B R^{-1} B^\top P_K) \\
   & = - 2 \operatorname{tr }\left(dP_K (A - B R^{-1} B^\top P_K )\right)\\
   & = - 2 \operatorname{tr }\left(\int_0^\infty e^{A_K^\top t}(U+U^\top) e^{A_K t}\mathrm{d }t \tilde{A}_K \right) \\
   & = - 4 \operatorname{tr }\left(\int_0^\infty e^{A_K^\top t} U e^{A_K t}\mathrm{d }t \tilde{A}_K\right). 
  \end{align}}By using the cyclic property and the linearity of the trace operator\footnote{By combining both properties, the equality $\operatorname{tr }(\int A(x) B(x) \mathrm{d}t) = \int \operatorname{tr }(A(x) B(x))  \mathrm{d}t=  \int \operatorname{tr }(B(x) A(x))  \mathrm{d}t= \operatorname{tr }(\int B(x) A(x) \mathrm{d}t)$ holds.}, we obtain
  {\small \begin{align}
   d e_K & = -4  \operatorname{tr }\left(U \int_0^\infty e^{A_K t} \tilde{A}_K e^{A_K^\top t}\mathrm{d}t  \right). \label{math:deK}
  \end{align}
  }Since the trace of a skew-symmetric matrix is always zero, only the symmetric part of the integral in \eqref{math:deK} contributes to the trace. 
  The symmetric part of the matrix $\int_0^\infty e^{A_K t} \tilde{A}_K e^{A_K^\top t}\mathrm{d}t $ is equal to the solution $X_K$ of Lyapunov equation 
  \begin{align}
   A_K X_K + X_K A_K^\top + \frac{1}{2}\left(\tilde{A}_K + \tilde{A}_K^\top \right)  =0. \label{math:Lyap4}
  \end{align}
  Combining everything yields
  \begin{align}
   d e_K = -4 \operatorname{ tr }\left( dK^\top \left(R K - B^\top P_K  \right) X_K \right). 
  \end{align}
  By using the relation $d f = \operatorname{ tr }(dZ W) \Rightarrow \nabla_Z f = W $ \cite{minka1997old}, we obtain \eqref{math:grad2}.
\end{proof}

\begin{remark}
   Note that the gradient \eqref{math:grad2} only exists for $K\in\mathcal{K}$. Otherwise, for $K\in \mathcal{K}_{\sigma} \setminus \mathcal{K}$, the integral in \eqref{math:deK} diverges because $A_K$ is not Hurwitz stable.
\end{remark}

\begin{theorem}
   \label{theorem:1}
   For $K\in\mathcal{K}$, the minimization problem of the Bellman error \eqref{math:sbe}, i.e.,
   \begin{align}
      \min_{K\in\mathcal{K}} e_K,
   \end{align}
   has a unique global minimizer $K^*=R^{-1}B^\top P^*$, where $P^*$ is the unique positive definite solution of the CARE \eqref{math:CARE}. Furthermore, $K^*$ is the only stationary point of $e_K$ on $\mathcal{K}$.
\end{theorem}

\begin{proof}
   The Bellman error $e_K$ is continuous (see Lemma~\ref{lemma:3}) and has nonempty compact sublevel sets on $\mathcal{K}$ (see Lemma~\ref{lemma:5}), given that $(A,B)$ is stabilizable (see Assumption~\ref{ass:1}). By the extreme value theorem \cite[Theorem 4.15.]{rudin1964principles}, a minimum exists on these sublevel sets. Since $e_K$ is not constant, the minimum lies within $L_{\alpha}$.
   To identify the stationary point, we prove that $X_K$ is nonsingular and solve $\nabla_K e_K=0$ with respect to $K$. First, we show that $\tilde{A}_K^\top  + \tilde{A}_K$ is nonsingular. 
   By substituting $A_K  = \tilde{A}_K + B R^{-1} B^\top  P_K - BK$ into \eqref{math:Ly2}, we obtain  

   {\small \vspace{-0.4cm}
   \begin{align}
      \tilde{A}_K^\top P_K + P_K \tilde{A}_K  + 2 P_K B R^{-1} B^\top P_K - K^\top B^\top P_K - P_K BK \notag  \\ = - Q - K^\top R K, \notag \\
   \Leftrightarrow \tilde{A}_K^\top P_K + P_K \tilde{A}_K = - Q - L^\top R L - (K+L)^\top R (K+L), \label{math:eq1}
   \end{align}}where $L=R^{-1}B^T P_K$. 
   The right-hand side of \eqref{math:eq1} is negative semidefinite and nonzero due to $Q\neq 0$ (see Assumption \ref{ass:1}). 
   Therefore, 
   \begin{align}
      v^T \left( \tilde{A}_K^\top P_K + P_K \tilde{A}_K \right)v \leq 0 \label{math:Lyap5}
      \end{align}
   must hold for all vectors $v\in\mathbb{R}^n$ and  $\tilde{A}_K^\top P_K + P_K \tilde{A}_K \neq 0$.
   Assume for a contradiction that $\tilde{A}_K^\top  + \tilde{A}_K$ is singular. Then, there exists a nonzero vector $v$ such that $(\tilde{A}_K^\top + \tilde{A}_K) v =0$. Substituting $\tilde{A}_K v = - \tilde{A}_K^\top v$ into \eqref{math:Lyap5} leads to
   \begin{align}
      v^T ( \underbrace{\tilde{A}_K^\top P_K - P_K \tilde{A}_K^\top}_{=:N}  ) v \leq 0. \label{math:Lyap6}
   \end{align}
   Now observe that 
   \begin{align}
      \operatorname{tr }(N )&= \operatorname{tr }(\tilde{A}_K^\top P_K) - \operatorname{tr }(P_K \tilde{A}_K^\top )\\
      & = \operatorname{tr }(\tilde{A}_K^\top P_K) - \operatorname{tr }( \tilde{A}_K^\top P_K ) =0.
   \end{align}
   Since the sum of the eigenvalues of $N$ is zero, $N$ cannot be negative semidefinite and nonzero, leading to a contradiction.
   Hence, $\tilde{A}_K^\top  + \tilde{A}_K$ is nonsingular. Next, we show that $X_K$ is nonsingular. Assume $X_K$ is singular, then there exists a nonzero vector $v$ such that $X_K v =0$. Pre- and post-multiplying \eqref{math:Lyap_A} by $v^\top$ and $v$, respectively, yields 
   \begin{align}
      0 &=  v^T \left( A_K X_K + X_K A_K^\top + \frac{1}{2} \left(\tilde{A}_K + \tilde{A}_K^\top  \right) \right) v \\
      & = v^T \frac{1}{2} \left(\tilde{A}_K + \tilde{A}_K^\top  \right) v. \label{math:Lyap7}
   \end{align}
   Since $\tilde{A}_K + \tilde{A}_K^\top $ is nonsingular, \eqref{math:Lyap7} implies $v=0$, which contradicts the assumption that $v$ is nonzero. Hence, $X_K$ is nonsingular.
   With $X_K$ nonsingular, the equation $\nabla_K e_K = 0$ implies $(RK - B^\top P_K)=0$. By setting $(RK - B^\top P_K)=0$, we obtain $K^*=R^{-1}B^\top P^*$, where $P^*$ is the unique positive definite solution of \eqref{math:CARE}.
    Therefore, $K^*$ is the unique minimizer of $e_K$ within $\mathcal{K}$.
\end{proof}

In the next theorem, we propose an optimization algorithm to solve the optimization problem of Theorem \ref{theorem:1}.

\begin{theorem}
   \label{theorem:2}
   The gradient flow 
   \begin{align}
      \dot K &= - \beta \nabla_K e_K, \quad K(0) = K_0 \in \mathcal{K}, \label{math:grad_flow}
   \end{align}
   where $\beta > 0$, solves the optimization problem of Theorem~\ref{theorem:1} and converges to the optimal feedback $K^*$. The trajectory $K(t)$ remains within $\mathcal{K}$ for all $t\geq 0$.
\end{theorem}
\begin{proof}
   Let $e_K$ be the Lyapunov functional. Since $M_K$ is negative semidefinite, $e_K =- \operatorname{tr }(M_K) = 0$ if and only if $M_K=0$, which implies $K=K^*$. Additionally, $e_K >0$ for $K\neq K^*$ over $\mathcal{K}$ because the codomain of $e_K$ is $\mathbb{R}_{\geq 0}$ (see Lemma \ref{lemma:M_K_def}). The time derivative of $e_K$ is given by
   \begin{align}
      \dot e_K  = \operatorname{tr }\bigl(\left(\nabla_K e_K\right)^\top \dot K \bigr) = - \alpha \Vert \nabla_K e_K \Vert^2_F \leq 0. \label{math:neg}
   \end{align}
   Since $\dot e_K <0 $ for all $K\in\mathcal{K}\setminus \{K^*\}$, $e_K$ is a proper Lyapunov function over $\mathcal{K}$ and the equilibrium point $K^*$ is asymptotically stable. Additionally, the trajectory $K(t)$ for $t>0$ remains in the sublevel set $L_{e_{K_0}}$ for any $K_0\in\mathcal{K}$, because $e_K$ is monotonically decreasing over time (see \eqref{math:neg}). Hence, the region of attraction of $K^*$ is $\mathcal{K}$.
\end{proof}

\begin{remark}
   The gradient flow of Theorem \ref{theorem:2} can be interpreted as a continuous-time analogue of a PI algorithm (see Algorithm \ref{alg:PI}), specifically Kleinman's Algorithm \eqref{math:Kleinman}. For a given policy $K$, solving the Lyapunov equation \eqref{math:Ly2} with respect to $P_K$ can be seen as the \textit{policy evaluation step}, where the value matrix is determined for the given policy. Simultaneously, the gradient flow updates the policy   by moving in the direction of the negative gradient, thereby reducing the Bellman error and improving the policy. This corresponds to the {policy improvement step} of the PI algorithm.
   By using the gradient flow, both steps occur concurrently in a continuous manner.
\end{remark}

\begin{remark}
 In contrast to the gradient flows in LQR theory, which minimize the true LQR cost, the proposed gradient flow minimizes the violation of an optimality condition, namely the HJB equation, quantified by the Bellman error. Since many RL methods are based on the HJB equation, 
 this work establishes a connection between continuous-time LQR and RL. We showed that the Bellman error can be parametrized with respect to the feedback gain, exhibits similar properties to the LQR cost function, and is thus well-suited for a policy gradient flow.
\end{remark}

\section{Simulation Results}
\label{sec:3} In this section, we illustrate the matrix function $e_K$ and its gradient flow, and  compare our method with the gradient flow of the LQR cost \eqref{math:grad} of \cite{bu2020clqr}. 
To visualize the results, we consider a system with dimensions $n=2$ and $m=1$ given by
\begin{align}
   \label{math:sys}
   A = \begin{bmatrix}-2 & 1 \\ 0 & -1 \end{bmatrix}, \quad B = \begin{bmatrix} 1 \\ 1\end{bmatrix}, \quad K = \begin{bmatrix} K_1 & K_2\end{bmatrix},
\end{align}
and $ Q = I_2, R = 2$. The system \eqref{math:sys} is stabilizable and the pair $(A,\sqrt{Q})$ is detectable. 
By using the Routh-Hurwitz stability criterion, the set of stabilizing feedbacks is given by
\begin{align}
   \mathcal{K} = \{K\in\mathbb{R}^{1 \times 2} : K_2 > -K_1 - 1   \}. \label{math:instable}
\end{align}
The Bellman error $e_K$ and the LQR cost $f_K$ are explicitly given by the rational functions\footnote{The explicit expressions are obtained by using the vectorization of the Lyapunov equations \eqref{math:Ly2} and \eqref{math:Lyap_A}.}
{\small
\begin{align*}
   e_K =   &  \big( K_1^6 + 4K_1^5K_2 + 12K_1^5 + 7K_1^4K_2^2 + 34K_1^4K_2 + 49K_1^4 \\
   & + 8K_1^3K_2^3 + 40K_1^3K_2^2 + 84K_1^3K_2 + 72K_1^3 \\
   & + 7K_1^2K_2^4 + 36K_1^2K_2^3 + 58K_1^2K_2^2 + 32K_1^2K_2 + 29K_1^2 \\
   & + 4K_1K_2^5 + 28K_1K_2^4 + 60K_1K_2^3 + 16K_1K_2^2 \\
   & - 52K_1K_2 - 8K_1 + K_2^6 + 10K_2^5 + 37K_2^4 \\
   & + 56K_2^3 + 17K_2^2 - 22K_2 + 5
   \big) / \\
    & \big(
   2\big(K_1^2 + 2K_1K_2 + 4K_1 + K_2^2 + 4K_2 + 3\big)^2
   \big),
\end{align*}
\begin{align*}
   f_K =  & 2\big( K_1^3 + 2 K_1^2 K_2 + 5 K_1^2 + 2 K_1 K_2^2 \\
   & + 4 K_1 K_2 + 4 K_1 + 2 K_2^3 + 7 K_2^2 + 2 K_2 + 5 \big) / \\
   & \big( 2 \left( K_1^2 + 2 K_1 K_2 + 4 K_1 + K_2^2 + 4 K_2 + 3 \right) \big).
\end{align*}}The Bellman error $e_K$ has a degree of $6$ and $4$ in the numerator and denominator, respectively, whereas the LQR cost $f_K$ has degrees of $3$ and $2$. The denominators of both functions are the same, except that one is squared, while the numerator polynomials differ.
Notably, the stability boundary $K_2 = -K_1 -1$ (see \eqref{math:instable}) corresponds to a root of the denominators of both $e_K$ and $f_K$.
\begin{figure}
   \centering
   \includegraphics[scale=0.75]{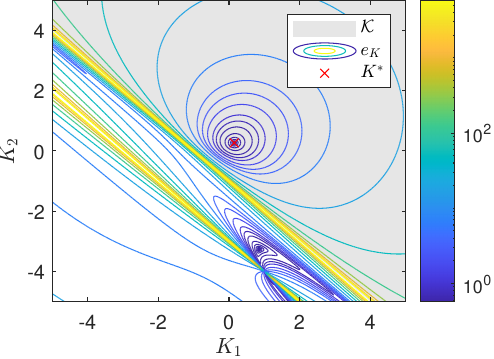}
   \caption{Contour plot of the Bellman error $e_K$}
   \label{fig:2}
\end{figure}
\begin{figure}
   \centering
   \includegraphics[scale=0.75]{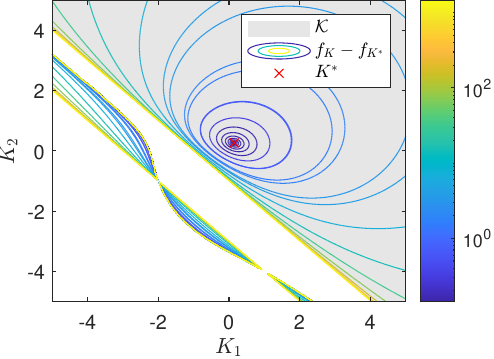}
   \caption[d]{Contour plot of the LQR cost $f_K-f_{K^*}$\footnotemark}
   \label{fig:1}
\end{figure}
\footnotetext{Although the integral of the LQR cost function diverges for unstable feedbacks, positive and finite costs can be obtained in the instability region $K_2 < -K_1 - 1$ by using the vectorization of the Lyapunov equations.}
The contour plots of $e_K$ and $f_K-f_{K^*}$ are shown in Figures \ref{fig:1} and \ref{fig:2}, where the same $28$ level sets are depicted. While the sublevel sets in the region $\mathcal{K}$ are similar, they exhibit different shapes. Both functions approach infinity by moving towards the stability boundary or as $\Vert K \Vert_2\to \infty$, $K\in\mathcal{K}$. Both matrix functions have a zero value at $K^*$ and are non-injective on $\mathcal{K}$, i.e., there are infinitely many policies for a given function value.

\begin{figure}
   \centering
   \includegraphics[scale=0.75]{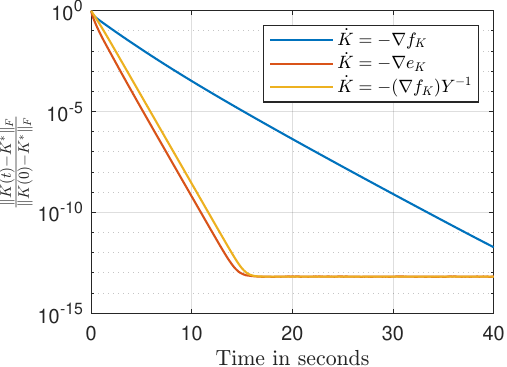}
    \caption{Normalized residuals $\frac{\Vert K(t) - K^* \Vert_F}{ \Vert K(0) - K^* \Vert_F }$ for (i) the gradient flow of the LQR cost, (ii) the gradient flow of the Bellman error, and (iii) the natural gradient flow of the LQR cost.}
   \label{fig:3}
\end{figure}

To evaluate convergence behavior, we generated 200 random instances with $A \in \mathbb{R}^{2 \times 2}$, $B \in \mathbb{R}^{2 \times 1}$, and a stabilizing initial gain $K(0) \in \mathbb{R}^{1 \times 2}$, where each matrix entry was sampled from a standard normal distribution. Figure~\ref{fig:3} shows the normalized residuals $\frac{\Vert K(t) - K^* \Vert_F}{ \Vert K(0) - K^* \Vert_F }$ over time for the gradient flow of the LQR cost, the gradient flow of the Bellman error, and the natural gradient flow of the LQR cost with $\gamma=1$. All three gradient flows exhibit linear convergence to the optimal gain $K^*$. The Bellman error gradient flow converges at a rate comparable to the natural gradient flow, while the standard gradient flow converges more slowly.

We observed that the gradient flow of the Bellman error may induce numerical instabilities, likely due to the presence of high-degree polynomial terms.
 In contrast, the natural gradient flow with $\gamma=1$ involves polynomials with lower degree, suggesting improved numerical stability. Simulations were performed using the \texttt{LSODA} ODE solver from \texttt{SciPy}, which provided reliable performance.

\section{Conclusion}
\label{sec:conc} 
In this paper, we introduced a novel continuous-time Bellman error for the LQR problem, parametrized by the feedback gain. We analyzed its properties, including its effective domain, smoothness, and coerciveness, and showed the existence of a unique stationary point. 
Furthermore, we derived a closed-form gradient expression of the Bellman error and established its well-definedness within the stabilizing policy region. 
The corresponding gradient flow solves the LQR problem and converges to the optimal feedback.
Additionally, this work provided novel insights on the LQR problem by bridging RL perspectives on the suboptimality of the CARE with first-order approaches in LQR theory.
We validated our method through simulations of multiple systems and demonstrated that our approach converges at a rate comparable to state-of-the-art methods.

\printbibliography

\end{document}